\documentclass[12pt,reqno]{amsart}
\usepackage{amsmath,amssymb,amsfonts,amsthm}
\usepackage[mathscr]{eucal}
\usepackage[russian,english]{babel}

\usepackage[all]{xy}
\usepackage{hyperref}
\title[Gauge symmetries in 2D field theory]
{Gauge symmetries in 2D field theory}
\author{S.L. Lyakhovich and  A.A. Sharapov}
\address{Department of Quantum Field Theory, Tomsk State University, Tomsk 634050, Russia}
\email{sll@phys.tsu.ru, sharapov@phys.tsu.ru}
\thanks{The work was partially supported by the project 2.3684.2011 of Tomsk State University and the RFBR grant 13-02-00551. A.Sh. appreciates the financial support from Dynasty Foundation.}

\newtheorem{thm}{Theorem}[section]
\newtheorem{prop}{Proposition}[section]
\newtheorem{cor}{Corollary}[section]

\def\deg{\mathrm{deg}}

\begin{document}

\maketitle

\begin{abstract}
A simple algorithm is proposed for constructing generators of gauge symmetry as well as reducibility relations for arbitrary systems of field equations in two dimensions.
\end{abstract}

\section{Introduction}
In this paper we suggest a simple and a general algorithm for finding
all gauge symmetries, given a system of local field equations in
two dimensions. The method works equally well for Lagrangian and
non-Lagrangian equations and it is  local in space-time.
In contemporary field theory, the field equations are often
constructed with a pre-specified gauge symmetry. In that case, one
has to be sure that the theory does not have any other
gauge symmetries. So, a systematic method of identifying a
\emph{complete} gauge symmetry of given field equations can be
useful even for the models having some known gauge invariance by
construction.

The Dirac-Bergmann algorithm allows one to find all gauge
symmetries for Lagrangian dynamics by casting the equations into the
normal form of the constrained Hamiltonian formalism \cite{Dirac},
\cite{GT}, \cite{HT}.  This algorithm can be extended to the general
systems, not necessarily Lagrangian, by bringing the dynamics to the
normal involutive form \cite{LS}. The Dirac-Bergmann algorithm was
originally formulated for mechanical systems. In this form, it has
been further developed by most of the followers, see for review
\cite{GT}, \cite{HT}, \cite{LS}. Its extension to field theory is
straightforward if the locality in space is not an
issue\footnote{For instance, the Dirac bracket, being an important
part of the Dirac formalism, may be non-local in space~\cite{GT},\cite{HT}.}. Besides locality,  the other subtleties are also known
concerning application of the classical Dirac-Bergmann algorithm
to field theories \cite{ST}.

The explicit knowledge of  space-time local  generators of a complete
gauge symmetry is a necessary pre-requisite for solving most of crucial
problems in field theory, like identifying global
symmetries and conservation laws, constructing  consistent
interactions and quantization \cite{BBH}. The recent developments in the BRST
formalism \cite{LS0}, \cite{KazLS}, \cite{LS1}, \cite{LS2} allow one
to solve the same range of problems for not necessarily Lagrangian
field theories. The list of examples of non-Lagrangian models of
current interest includes chiral bosons in various dimensions,
Seiberg-Witten and Donaldson-Uhlenbeck-Yau equations, various
conformal field theories with extended supersymmetry, and
M.A.~Vasiliev equations of interacting higher-spin massless fields.

While the importance of explicit identification of  gauge
symmetries is widely recognized in physics, on the mathematical side
the gauge invariance of PDEs is often considered as an ``unpleasant
complication'', which should be overcome immediately by imposing
appropriate gauge fixing conditions making the system fully
determined (see e.g. \cite{VK}, \cite{Seiler}). Perhaps, the only
exception to this practice is the mathematical theory of optimal
control, where the gauge symmetry reincarnates as controlability. An
expanded discussion of the relationship between both the concepts
can be found in \cite{LS}. In that paper, we also described a normal
form that the general system of ODEs can be brought into, and proved some
basic theorems on the structure of gauge symmetry transformations.

The present paper extends the results of \cite{LS} to the general,
not necessarily Lagrangian, 2D field theory, providing a systematic
method for finding a complete gauge symmetry. The
extension is not straightforward due to appearance of new
integrability conditions steaming from commutativity of partial
derivatives, that has no analogue in mechanics. The main difference,
however, is the change of the ground ring underlying the analysis of
gauge symmetries.  The situation can be   described schematically by
the following table:
$$
\begin{tabular}{|c|c|c|c|}
  \hline
   & ODEs : D=1 & PDEs , D=2 & PDEs , D$>$2 \\ \hline
  The ground ring & meromorphic functions                & ordinary diff. operators & partial diff. operators \\
              &                          (m.f.)           & with coefficients in m.f. & with coefficients in m.f.\\
  \hline
  Algebraic  & commutative,        & non-commutative,  & non-commutative, \\
properties   & differential field  & principal ideal domain & Noetherian \\
  \hline
\end{tabular}
$$
As is seen, the 2D field theories are intermediate in algebraic properties between ODEs and higher dimensional PDEs.
This allows us to consider the case of two dimensions
as special\footnote{Let us also mention a plenty of nonlinear integrable models known in D=2,
though this fact is not directly related to the present work.}.

The structure of the present paper corresponds to the structure of
the algorithm we propose for finding gauge symmetries. The
latter includes three steps. Given a system of 2D PDEs, we transform
it to the Cartan normal form, revealing thus all hidden
integrability conditions, if any. As a result we get a formally
integrable system of the first order PDEs. This preparatory step is
quite standard and it is explained in Sec. 2. In the same section,
we also recall an algebraic background needed for a rigorous
definition of the notion of a gauge symmetry and illustrate this
notion by two simple yet general examples. These examples
demonstrate
In Sec. 3.1, we show that any 2D field theory in the
Cartan normal form can be embedded into a constrained Hamiltonian
system, which we call the Pontryagin system, in such a way that the
gauge symmetries of the original equations extend to those of the
Pontryagin action. Applying the second Noether theorem to the
Pontryagin system reduces the problem of finding gauge symmetries to
that of finding the Noether identities, as it is explained in Sec.
3.2. Due to the special structure of the Hamiltonian equations, the
latter problem amounts to constructing differential identities among
the primary Hamiltonian constraints and it is the point where the
theory of finitely generated modules over rings of differential
operators comes to forefront. In Sec. 3.3, we construct a minimal
free resolution for the differential module associated with the
Hamiltonian constraints, from which both a generating set for the
gauge symmetry transformations and the corresponding reducibility
relations can be read off. Among other things, this construction
provides a direct proof of the fact that in 2D field theory any
gauge symmetry admits no more than one stage of reducibility.

In Sec. 4, we consider a particular example of nonlinear relativistic field equation. As well as being an  illustration to our method,  it demonstrates an interesting phenomenon of bifurcation of the structure of gauge symmetry when one varies numerical parameters entering the model.
In particular, it shows that a smooth deformation of free field equations by inclusion of interaction is not always followed by a smooth deformation of the corresponding gauge  generators, even though the overall number of independent gauge symmetries is preserved.

In the concluding Sec. 5, we summarize our results and formulate two plausible conjectures about the count of physical degrees of freedom in 2D field theory. The Appendix contains a useful theorem on the matrices over the ring of ordinary differential operators.

\section{Cartan normal form and gauge symmetries}
 By a two-dimensional field theory we understand an arbitrary
system of PDEs with two in\-de\-pen\-dent variables. We fix
neither the order of equations nor their number, which may
be completely arbitrary and in no way correlate with the number of
dependent variables (fields). In this section, we discuss a normal
form each two-dimensional system of field equations  can be brought into at the
cost of introducing axillary fields. This normal form will be a
starting point for the study of gauge symmetries in the next
section.

\subsection{Pfaffian systems} Let $\Lambda(M)=\bigoplus
\Lambda^k(M)$ denote the exterior algebra of differential forms on
an $n$-dimensional manifold $M$ and let $\mathcal{I} \subset
\Lambda(M)$ be an ideal of $\Lambda(M)$. A submanifold
$\Sigma\subset M$ is called an integral manifold of $\mathcal{I}$
if $\alpha|_\Sigma=0$ for all $\alpha \in {\mathcal{I}}$. In the
case where the ideal $\mathcal{I}$ is generated by a set of
$1$-forms $\Theta^J$ and $0$-forms $\Phi_A$ the looking-for
integral manifolds is known as the Pfaff problem. The
corresponding system of equations defining the integral manifolds,
\begin{equation}\label{}
    \Theta{}^J|_{\Sigma}=0\,,\qquad \Phi_A|_\Sigma=0\,,
\end{equation}
is called the Pfaffian system.

Let us  indicate how any system of PDEs
can be reduced to a Pfaffian system. If the system contains
equations of order higher than one, we can first reduce it to the
order one by introducing new unknown variables which represent
certain derivatives of the original ones. This being done, we
obtain a system of equations of the form $\Phi_A(x^i, \phi^J,
\partial \phi^J /\partial x^i)=0$, where $\{\phi^J\}_{J=1}^{n}$ are
the unknown and $\{x^i\}_{i=1}^d$ are the independent variables.
If we set $\varphi^J_i=\partial \phi^J/\partial x^i$, the original
system of PDEs can be replaced by the
Pfaffian system composed of the equations
$\Phi_A(x,\phi,\varphi)=0$ and $\Theta^J\equiv d
\phi^J-\varphi^J_idx^i=0$. The solutions of the original system
correspond to those solutions of the Pfaffian system that are
integral manifolds  of dimension $d$ on which the variables $x^i$
are independent. The last condition can be written as
$$dx^1\wedge\cdots\wedge dx^d|_{\Sigma}\neq 0\,.$$

Notice that the equations $\Phi_A=0$ define a submanifold
$N\subset M$ (perhaps with singularities) so that any integral
manifold $\Sigma$ belongs to  $N$. Therefore, without loss in
generality, we can restrict the $1$-forms $\Theta^J$ on $N$ and
obtain an equivalent Pfaffian system on $N$ generated by the
$1$-forms $\theta^J=\Theta^J|_N$. Actually, it is the system of
equations
\begin{equation}\label{PS}
    \theta^J|_\Sigma=0
\end{equation}
that is usually referred to as  a Pfaffian system (no algebraic
constraints $\Phi_A=0$).

Let $\omega^J = d\theta^J\in \Lambda^2(N)$.  Since the operations
of restriction and exterior differentiation commute to each other,
we have
\begin{equation}\label{dPS}
    \omega^J|_\Sigma=0
\end{equation}
whenever $\Sigma$ is a solution to  the Pfaff problem (\ref{PS}).
Therefore,  it is reasonable to consider equations (\ref{PS}) and
(\ref{dPS}) together. The procedure of adjoining to a Pfaffian
system the exterior differentials of its $1$-forms is just an
invariant way to allow for all possible integrability  conditions
associated with the original system of PDEs. Taken together, the
1-forms $\theta^J$ and 2-forms $\omega^J$ generate a differential
ideal $\mathcal{J} \subset \Lambda(N)$ called usually  an
\textit{exterior differential system} on $N$.

\subsection{Cartan normal form} In what follows we will exclusively deal with
the two-dimensional field theory. In view of the above this is
equivalent to the study of two-dimensional integral manifolds
$\Sigma$ for the Pfaffian system (\ref{PS}). Since our
consideration will be essentially local, we may assume the
manifold $N$ - the target space of fields - to be a suitable open
domain in $\mathbb{R}^n$ with Cartesian coordinates $\phi^i$, while
the surface $\Sigma$ - the source space of fields - is a
two-dimensional domain with coordinates $x$ and $\bar x$.
Restricting the target space, if necessary, we may further assume
the $1$-forms $\theta^J$, $J=1,\ldots, m$, to be linearly
independent at each point, that is, $\theta^1\wedge\cdots\wedge
\theta^m \neq 0$. Then, locally, we can separate the coordinates
$\phi$'s into two groups $\phi^J$ and $\phi^a$ and rearrange the
basis of 1-forms $\theta$'s in such a way that the Pfaffian system
takes the form
\begin{equation}\label{theta}
    \theta^J=d\phi^J-  Z_a^J(\phi) d\phi^a\,.
\end{equation}
The algebraic ideal generated by $\theta$'s can be extended  to
the differential ideal by adjoining the $2$-forms $d\theta^J$. A
straightforward computation yields
\begin{equation}\label{dt}
    d\theta^J=\Omega^J_{ab}(\phi)d\phi^a\wedge
    d\phi^b \qquad (\mathrm{mod} \;\theta)\,,
\end{equation}
where
\begin{equation}\label{}
   \Omega_{ab}^J= \partial_a
    Z_b^J-\partial_b
    Z_a^J+Z_a^I\partial_IZ^J_b-Z_b^I\partial_IZ^J_a\,.
\end{equation}
Clearly, the fields $\phi^i(x,\bar x)$ define  an integral surface
$\Sigma\subset \mathbb{R}^2$ for the exterior differential system
(\ref{theta}), (\ref{dt}) iff the following system of PDEs is
satisfied:
\begin{equation}\label{3eq}
T^J=\partial \phi ^J-Z_a^J(\phi) \partial \phi^a=0 \,,\qquad
\overline{T}{}^J=\bar\partial \phi^J-Z_a^J(\phi) \bar\partial
\phi{}^a=0\,,\qquad \widehat{T}^J=\Omega^J_{ab}(\phi)\partial
\phi^a
    \bar\partial \phi^b=0\,.
\end{equation}
These equations are not independent. As a consequence  of
(\ref{dt}) we have the identities
\begin{equation}\label{NId}
   \bar\partial T^J -T^I\partial_IZ_a^J\bar\partial \phi^a -\partial \bar T{}^J + \overline{T}{}^I\partial_IZ_a^J\partial \phi^a - \widehat{T}{}^J
   \equiv 0  \,.
\end{equation}
Let us interpret the equation $\widehat{T}{}^J=0$ as a system of
linear homogeneous equations with respect to the unknowns $\partial
\phi^a$. Then the general solution to this system can be written
as
\begin{equation}\label{lambda}
    \partial \phi^a=Z^a_\alpha(\phi,\bar\partial \phi)\lambda^\alpha\,,
\end{equation}
where the $\lambda$'s are arbitrary functions of $x$ and $\bar x$.
The number of the new fields $\lambda^\alpha$ is equal to $n-m-l$,
where $l$ is the rank of the matrix
$(\bar \partial\phi^b\Omega^J_{ba})$ in general position\footnote{
Notice that $n-m-l\geq 1$ as we always have an obvious solution
$\partial \varphi^a=\lambda\bar\partial\phi^a$. In case $n-m-l=1$
the integral surface $\Sigma$ degenerates into a curve since the
tangent vectors $\partial\phi^i$ and $\bar\partial\phi^i$ become
linearly dependent.}. System (\ref{3eq}) is now equivalent to
the following one:
\begin{equation}\label{ttt}
\begin{array}{l}
\widetilde{T}{}^J=\partial \phi^J-Z_a^J(\phi)
Z^a_\alpha(\phi,\bar\partial
\phi)\lambda^\alpha=0\,,\\[5mm]
T^a=\partial \phi^a-Z^a_\alpha(\phi,\bar\partial
\phi)\lambda^\alpha=0\,,\\[5mm]  \overline{T}{}^J=\bar\partial \phi^J-Z_a^J(\phi)
\bar\partial \phi^a=0\,.
\end{array}
\end{equation}
The identities (\ref{NId}) take the form
\begin{equation}\label{GID}
\partial  \overline{T}{}^J=\bar\partial
{\widetilde{T}}^J-T^a\Omega_{ab}^J\bar\partial \phi^b  -\bar\partial
(Z_a^J T^a)\,.
\end{equation}

Treating the first two equations in (\ref{ttt}) on equal footing,
we arrive at the following normal form of PDEs describing a
two-dimensional field theory:
\begin{equation}\label{NF}
    T^i\equiv \partial \phi^i-Z^i_\alpha(\phi,\bar\partial
    \phi^a)\lambda^\alpha=0\,,\qquad \overline{T}{}^J\equiv\bar\partial \phi^J-Z_a^J(\phi)
\bar\partial \phi^a=0\,.
\end{equation}
Here we also used the third equation in (\ref{ttt}) to express the
derivatives  $\bar\partial \phi^J$ in $\tilde{T}^J$ and $T^a$ in terms of
$\phi^i$ and $\bar\partial \phi^a$. It is convenient to think of
the independent variables $x$ and $\bar x$ as  the time and space
coordinates, respectively. Then the first equation in (\ref{NF})
governs the time evolution, while  the second one imposes
constraints on the space derivatives of fields. In
view of the  identity (\ref{GID}),  the constraint surface is
preserved by the time evolution.

Although any system of PDEs on plane can locally be reduced to the
normal form (\ref{NF}), the reduction can lead to a considerable increase in the size of the system. For this reason, it
is useful to slightly relax the form of equations (\ref{NF}) by
allowing higher-order space derivatives of fields together with nonlinear dependence of $\lambda$'s.
This leads us to what is known as the \textit{Cartan normal form} of equations:
\begin{equation}\label{RNF}
\begin{array}{l}
 T^i\equiv \partial \phi^i-Z^i(\phi,\bar\partial
    \phi^a, \ldots
    , \bar\partial^q
    \phi^a; \lambda, \bar\partial \lambda^\alpha,\ldots, \bar\partial^p\lambda^\alpha)=0\,,\\[3mm] \overline{T}{}^J\equiv\bar\partial
    \phi^J-Z^J(\phi,\bar\partial \phi^a, \ldots, \bar\partial^p \phi^a) =0\,,\\[3mm]i=1,\ldots, n\,,\qquad  \alpha=1,\ldots, l\,, \qquad J=1,\ldots, m\,,\qquad a=1,\ldots, n-m\,.
    \end{array}
\end{equation}
It is implied that the differential  constraints and the evolutionary
equations still satisfy  the compatibility condition
\begin{equation}\label{CC}
 \partial \overline{T}{}^J=U_I^J \overline{T}{}^I+V_{i}^JT^i\,,
\end{equation}
where $$U_I^J=\sum_n U_{I n}^J\bar\partial^n\,,\qquad V_i^J=\sum_n V_{in}^J\bar\partial^n$$ are  matrix differential operators in $\bar x$
with coefficients depending on fields and their derivatives.
Like (\ref{GID}),  the condition (\ref{CC}) ensures stationarity of the constraint surface $\overline{T}{}^J=0$.

As we will see in the next section,
replacing the Cartan normal form (\ref{RNF}) with the more rigid one
(\ref{NF}) yields no material simplification.

Having brought the equations into the Cartan normal form, one can
easily prove the existence and uniqueness of their solutions under
the assumption of analyticity of $Z$'s. In the Cartan approach the
integration procedure includes two steps: first one defines
admissible Cauchy data at a given instant of time and then integrate
the evolutionary equations. In more detail, the construction goes as
follows. Let $(x_0,\bar x_0)$ be an arbitrary space-time point.
Choose $n-m$ real-analytic functions $\phi^a(\bar x)=\phi^a(x_0,\bar
x)$ of $\bar x$. Substituting these functions into the second
equation in (\ref{RNF}) yields a well-defined system of $m$ ordinary
differential equations for the unknowns  $\phi^J(\bar
x)=\phi^J(x_0,\bar x)$. The equations have a unique solution subject
to the initial condition $\phi^J(\bar x_0)=\phi^J_0$. The curve
$\phi^i(x_0, x)=(\phi^a(\bar x), \phi^J(\bar x))$ is then used as
the Cauchy data for the first equation in (\ref{RNF}). Again, as
with the differential constraints, the evolutionary equations are in
the underdetermined Kovalevskaya form. This means that we can
prescribe $\lambda$'s to be any real-analytic functions of $x$ and
$\bar x$. Once these functions have been specified, the equations
$T^i=0$ take the usual Kovalevskaya form and the famous
Cauchy-Kovalevskaya theorem ensures the existence of a unique
solution $\phi^i(x,\bar x)$ with the initial data $\phi^i(x_0, x)$.
By construction, this solution satisfies the equation
$\overline{T}^J=0$ at $x=x_0$,  and hence for all $x$'s due to the
compatibility condition (\ref{CC}). Thus, we see that the general
solution to (\ref{RNF}) is determined by $m$ constants $\phi^J_0$,
$n-m$ analytic functions $\phi^a(\bar x)$ of a single variable, and
$l$ analytic functions $\lambda^\alpha(x, \bar x) $ of two
variables.

\subsection{Gauge symmetries}
Dependence of the general solution of the arbitrary analytic
functions $\lambda^\alpha$ suggests that the system (\ref{RNF}) enjoys
an $l$-parameter gauge symmetry. By a gauge symmetry
we understand an infinitesimal transformation
\begin{equation}\label{GT}
    \delta_\varepsilon \phi^i=\sum_{q,p=0}^{Q,P}
    R^i_{qp}\partial^q\bar\partial^p\varepsilon\,,\qquad
    \delta_\varepsilon
    \lambda^\alpha=\sum_{q,p=0}^{Q,P}R^\alpha_{qp}\partial^q\bar\partial^p\varepsilon\,
\end{equation}
that leaves invariant the field equations (\ref{RNF}). Here the
gauge parameter $\varepsilon$ is assumed to be an arbitrary
function of $x$ and $\bar x$,  and the coefficients $R$'s are functions of fields $\phi^i$, $\lambda^\alpha$ and their
derivatives up to some finite order. Invariance of the field
equations means that for any choice of $\varepsilon$ one has
\begin{equation}\label{deltaT}
    \delta_{\varepsilon}T^i \approx 0\,,\qquad \delta_\varepsilon
    \overline{T}{}^J \approx 0 \,,
\end{equation}
where the sign $\approx$ means ``modulo equations of motion (\ref{RNF}) and their differential consequences''. We borrow this notation from the constrained dynamics \cite{Dirac}, \cite{HT}.

The number $Q+P$ is called the \textit{order of gauge symmetry}
if the coefficients $R^i_{QP}$ and $R^\alpha_{QP}$ are not all
equal to zero identically.  A gauge symmetry (\ref{GT}) is called
\textit{trivial} if $R^i_{qp}\approx 0$, $R^\alpha_{qp}\approx 0$,
i.e., if it has no effect upon any solution to the field
equations. The trivial gauge symmetries are present in any field theory without any material consequences.  This motivates us to define the space of \textit{nontrivial gauge symmetries} $\mathcal{G}$ as the quotient space of all gauge symmetries by the trivial ones.
A precise definition of the space $\mathcal{G}$ will be given below, but before going into details we would like to present a pair of quite general examples of the first-order gauge symmetries.

\vspace{3mm}
\noindent \textit{Example 1.} Notice that the original field equations
(\ref{PS}) are invariant under diffeomorphisms of the
integral surface $\Sigma$.  This yields  the following gauge
transformations:
\begin{equation}\label{rep}
    \delta_{\epsilon}\phi^i=\mathcal{L}_\epsilon \phi^i=\varepsilon\partial \phi^i+\bar\varepsilon\bar\partial\phi^i\,,
\end{equation}
where $\epsilon=\varepsilon(x,\bar
x)\partial+\bar\varepsilon(x,\bar x)\bar\partial$ is an arbitrary
infinitesimal vector field on $\Sigma$. These transformations can
be easily extended to the Cartan normal form (\ref{NF}). For this
end, one needs only to express the fields $\lambda^\alpha$ from
(\ref{NF}) as functions of $\phi^i$, $\partial\phi^i$, and $\bar\partial
\phi^a$. Varying the resulting expression, one then obtains the gauge
transformation $\delta_\epsilon \lambda^\alpha$ as a linear combination
of (\ref{rep}). Clearly, the value $\delta_\epsilon
\lambda^\alpha$ involves no more than the first partial
derivatives of the gauge parameters  $\varepsilon$ and
$\bar\varepsilon$. Of
course, we do not claim that the two-parameter transformation
(\ref{rep}) exhausts all gauge symmetries of the field equations
(\ref{PS}).

\vspace{3mm} \noindent \textit{Example 2.} Consider a \textit{completely integrable}
Pfaffian system (\ref{PS}). In this case, the $1$-forms $\theta^J$
generate a differentially closed ideal $\mathcal{J}$, so that
$\Omega_{ab}^J=0$ and the equations (\ref{NF}) take the
form\footnote{Clearly, one can omit the last equation without serious
consequences, as it just expresses the auxiliary fields
$\lambda$'s in terms of the original fields $\phi$'s.}
\begin{equation}\label{}
\partial \phi ^J=Z_a^J(\phi) \partial \phi^a \,,\qquad
\bar\partial \phi^J=Z_a^J(\phi) \bar\partial \phi{}^a\,,\qquad
\partial \phi^a=\lambda^a\,.
\end{equation}
It is easy to see that the system enjoys  the following gauge
symmetries:
\begin{equation}\label{GT2}
    \delta_\varepsilon \phi^a=\varepsilon^a \,,\qquad
    \delta_{\varepsilon}\phi^J=Z^J_a(\phi)\varepsilon^a\,,\qquad
    \delta_\varepsilon \lambda^a=\partial \varepsilon^a\,.
\end{equation}
Since the number of gauge parameters coincides with the number of
$\lambda$'s, one can expect that these transformations exhaust all
the gauge symmetries of the system. In the next section,  we will
show that this is so indeed. The gauge transformations (\ref{GT2})
have the following geometrical origin. The Pfaffian system, being
integrable, defines an ($n-m$)-dimensional foliation
$\mathcal{F}(N)$ of the target space. By definition,
$\theta^J|_S=0$ for any leaf $S\in \mathcal{F}(N)$ and the leaves
of $\mathcal{F}(N)$ have the maximal possible dimension among the integral
manifolds of $\theta$'s. If $n-m\geq 2$, then each two-dimensional
integral manifold $\Sigma$ has to belong to some leaf $S$ and the
gauge symmetries (\ref{GT2}) are simply induced  by the
diffeomorphisms of $S$. In particular, these diffeomorphisms
absorb the diffeomorphisms of the submanifold $\Sigma\subset N$.
The last fact allows one to write the gauge transformations
(\ref{rep}) as a specialization of (\ref{GT2}) for $\varepsilon^a=
\mathcal{L}_\epsilon\phi^a\approx \varepsilon
\lambda^a+\bar\varepsilon\bar\partial\phi^a$.

\vspace{3mm}
Now let us give a formal definition of the space of nontrivial gauge symmetries $\mathcal{G}$. This will require some algebraic background and terminology. First, we define the ring ${\mathfrak{A}}$ constituted by the real-analytic functions of \textit{finite} number of variables $\partial^q\bar\partial ^p\lambda^\alpha$ and $\partial^q\bar\partial^p\phi^i$. Since $\mathfrak{A}$ is an integrality domain,  we can introduce the field of quotients $\mathfrak{F}=\mathrm{Quat}(\mathfrak{A})$. The natural action of the partial derivatives $\partial$ and $\bar\partial$ makes $\mathfrak{F}$ into a differential field.
Denote by $ \mathfrak{R}=\mathfrak{F}[\partial,\bar \partial]$ the noncommutative ring of differential operators with coefficients in $\mathfrak{R}$. The general element of $\mathfrak{R}$ reads
$$
A=\sum_{n,m=1}^{N,M}{A}_{nm}\partial^n\bar\partial^m \,,\qquad A_{nm}\in {{\mathfrak{F}}}\,.
$$
The ring $\mathfrak{R}$ is known to be simple and Noetherian. We write $\mathfrak{R}^{n\times {m}}$ for the set of $n\times m$-matrices with entries in $\mathfrak{R}$.

Let ${\mathfrak{I}}$ denote the \textit{differential} ideal generated  by the left hand sides of the field equations $T^i$ and $\overline{T}^J$.  Since the equations (\ref{RNF}) are solved for $\partial\phi^i$ and $\bar\partial\phi^J$, the ideal ${\mathfrak{I}}$ is prime. As a consequence the quotient ring ${\mathfrak{R}}/{\mathfrak{I}}$ is an integrality domain. Replacing in the above definitions $\mathfrak{R}$ by $\mathfrak{R}/\mathfrak{J}$,  we  define the differential field of quotients ${{\mathcal{F}}}=\mathrm{Quot}({\mathfrak{R}}/{\mathfrak{I}})$, the noncommutative  ring of differential operators $\mathcal{R}=\mathcal{F}[\partial,\bar\partial]$, and the  ${\mathcal{R}}$-module $\mathcal{R}^{n\times m}$ of $n\times m$-matrices over ${\mathcal{R}}$.
We will identify the right (left)  $\mathcal{R}$-module $\mathcal{R}^n$ with $\mathcal{R}^{n\times 1}$ (resp. $\mathcal{R}^{1\times n}$) and refer to its elements as vectors (resp. covectors). Like $\mathfrak{R}$,  the ring $\mathcal{R}$ is Noetherian, so that any submodule of the free module $\mathcal{R}^n$ is finitely generated.

It follows from the definition of $\mathfrak{J}$ that any element of $\mathcal{F}$ can be uniquely represented by a  real-meromorphic function of $\phi^i$, $\bar{\partial}^q\phi^a$, and $\partial^q\bar\partial^p\lambda^\alpha$. This allows us to identify $\mathcal{F}$ with a subfield of $\mathfrak{F}$ and  $\mathcal{R}$ with a subring of $\mathfrak{R}$. In what follows these identifications will be always implied.

  Now we are ready to define the space of infinitesimal gauge transformations. Consider the universal linearization \cite{VK} of the  field equations (\ref{RNF}). It is obtained by extending the field equations (\ref{RNF}) with their variations
$$
\delta T^i=E^i_j\delta\phi^j+E^i_\alpha \delta\lambda^\alpha=0\,,\qquad \delta \overline{T}{}^J=E^J_j\delta\phi^j=0\,.
$$
Here $\delta \phi^i$ and $\delta \lambda^\alpha$ are regarded  as new unknowns. The coefficients of the new equations can be combined into a single matrix
$$
{\mathbf{E}}=\left(
    \begin{array}{cc}
      E^i_j & E^i_\alpha \\
      E^J_j & 0 \\
    \end{array}
  \right)\in \mathfrak{R}^{(n+m)\times (n+l)}\,.
$$
With account of the original field equations (\ref{RNF}), we can think of this matrix as representing an element of ${\mathcal{R}}^{(n+m)\times (n+l)}$. Similarly, the coefficients of the gauge transformation (\ref{GT}) define the  vector
\begin{equation}\label{R}
{R}=\left(
    \begin{array}{c}
      R^i \\
      R^\alpha \\
    \end{array}
  \right)\in {\mathcal{R}}^{n+l}\,,
\end{equation}
which is called the \textit{generator of a gauge transformation}. The defining condition (\ref{deltaT}) for $R$ to be a gauge symmetry generator is then equivalent to the relation ${\mathbf{E}}{R}=0$. The solutions to the last equation form a submodule of the right ${\mathcal{R}}$-module ${\mathcal{R}}^{n+l}$ and we identify this submodule with the space of gauge symmetries $\mathcal{G}$.

The ring ${{\mathcal{R}}}$ being Noetherian, the submodule of gauge trans\-for\-mations ${\mathcal{G}}\subset {\mathcal{R}}^{n+l}$ is finitely generated. Let $R_1, R_2, \ldots, R_{r_1}$ be a finite set of vectors generating  $\mathcal{G}$. We can arrange these vectors into the matrix $$
 {\mathbf{R}}=(R_1, R_2,\ldots, R_{r_1})\in {\mathcal{R}}^{(n+l)\times r_1}\,,
 $$
 so that $\mathbf{E}\mathbf{R}=0$.  Following the physical tradition, we will refer to  $\mathbf{{R}}$ as complete set of gauge symmetry generators. Given a complete set of generators,  any gauge transformation $R\in \mathcal{G}$ can be written as
$
R={\mathbf{R}}K
$ for some $K\in {\mathcal{R}}^{r_1}$. In general, it is impossible to choose the generators $R_1, R_2,\ldots, R_{r_1}$ of $\mathcal{G}$ in a linearly independent way, that is, for any choice of ${\mathbf{R}}$ there may exist a nonzero vector $Z\in {{\mathcal{R}}}^{r_1}$ such that ${\mathbf{R}}Z=0$. The null-vectors of ${\mathbf{R}}$ form a right ${\mathcal{R}}$-module $\mathrm{Syz}({\mathbf{R}})$ called the module of first syzygies. Again the syzygy module, being a submudule of a Noetherian module,  is finitely generated and one can arrange its generators into an $r_2\times r_1$-matrix ${\mathbf{Z}}$. By definition, ${\mathbf{R}}{\mathbf{Z}}=0$ and for any $Z\in \ker {\mathbf{R}}$ there exists $L\in {\mathcal{R}}^{r_2}$ such that $Z={\mathbf{K}}L$. According to the terminology adopted in the physical literature the gauge symmetries with $\mathrm{Syz}({\mathbf{R}})\neq 0$ are called \textit{reducible} and the columns  of the matrix ${\mathbf{Z}}$ are referred to as the generators of \textit{reducibility relations}.
It may happen that $\ker {\mathbf{Z}}\neq 0$, that is, the generators of reducibility relations are reducible themselves.  Then one can define the right ${\mathcal{R}}$-module $\mathrm{Syz}(\mathbf{Z})$ of second syzygies, which is also finitely generated. Iterating this construction once and again yields a free resolution of the right ${\mathcal{R}}$-module $\mathcal{G}$:
$$
\xymatrix{\cdots \ar[r]&{\mathcal{R}}^{r_2}\ar[r]^{{\mathbf{Z}}}&{\mathcal{R}}^{r_1}\ar[r]^{{\mathbf{R}}}&\mathcal{G}\ar[r]&0 }\,.
$$
In principle, the chain of syzygy modules may continue to infinity, but the general theorems on the ring of differential operators ensure the existence of a finite resolution. The minimal possible length of a finite resolution is called the global dimension of the module $\mathcal{G}$. For $D=2$ the global dimension of any finitely generated module over the ring of differential operators with coefficients in differential field is known to be bounded by 2. Hence, there exists a short exact sequence of modules\footnote{The case $r_2=0$ and $\mathbf{Z}=0$ is not excluded.}
\begin{equation}\label{MR}
\xymatrix{0 \ar[r]&{\mathcal{R}}^{r_2}\ar[r]^{{\mathbf{Z}}}&{\mathcal{R}}^{r_1}\ar[r]^{{{{\mathbf{R}}}}}&\mathcal{G}\ar[r]&0 }\,.
\end{equation}
Since the map $\mathbf{Z}$ is injective, the rows of the matrix $\mathbf{Z}$ form a basis in the space of reducibility relations for the complete set of gauge symmetry generators ${\mathbf{R}}$.

Thus, the problem of finding  gauge symmetries for a given set of equations amounts to solving  linear equations over the ring ${\mathcal{R}}$. The latter problem admits, in principle, an algorithmic solution by means of various algebraic techniques exploiting the idea of Gr\"obner bases. Our strategy, however, will be somewhat different: instead of considering  an abstract system of two-dimensional field equations we would like  to take an advantage of the Cartan normal form. As a matter of fact, the use of Cartan normal form allows us to much extent replace the study of linear equations over the ring of partial differential operators ${\mathcal{R}}={{\mathcal{F}}}[\partial, \bar\partial]$ to a similar problem for the subring of ordinary differential operators $\bar{{\mathcal{R}}}={{\mathcal{F}}}[\bar\partial]$. The latter ring enjoys special algebraic properties that considerably simplify  computation of the gauge symmetry generators as well as the reducibility relations.
Contrary to  $\mathcal{R}$, the ring of ordinary differential operators $\bar{\mathcal{R}}$ is a
left and a right principal ideal domain, meaning that every left and every right ideal can be generated by one single element. Furthermore, $\bar{\mathcal{R}}$  is a left and a right Euclidean domain, which means that we have the left and right ``division with remainder'' with respect to the order of ordinary differential operators. All these properties make the linear algebra over  $\bar{\mathcal{R}}$ to be somewhat similar to that of vector spaces. In particular, we will extensively use the following fundamental result
\begin{thm}\label{Th}
Any submodule of the free module $\bar{\mathcal{R}}^n$ is a free module of rank $\leq n$.
\end{thm}
 In other words, the submodules of $\bar{\mathcal{R}}^n$ behave like the subspaces of an $n$-dimensional vector space. In particular, each submodule admits a finite basis.  There is, however, one striking  difference: although any submodule $\mathcal{M}\subset \bar{\mathcal{R}}^n$ is isomorphic  to a free module $\bar{\mathcal{R}}^m$ with $m\leq n$, the quotient module $\bar{\mathcal{R}}^n/\mathcal{M}$ may not be free. In general, $\bar{\mathcal{R}}^n/\mathcal{M}\simeq \bar{\mathcal{R}}^{n-m}\oplus \mathcal{T}$, where $\mathcal{T}$ is a torsion module. Stated differently, a submodule $\mathcal{M}\subset \bar{\mathcal{R}}^n$ may not admit a complimentary submodule, that is, a submodule $\mathcal{N}\subset \bar{\mathcal{R}}^n$ such that $\bar{\mathcal{R}}^n=\mathcal{M}\oplus \mathcal{N}$. It is the presence of torsion (or absence of a complimentary submodule) that makes the main difference between the submodules of $\bar{\mathcal{R}}^n$ and the vector subspaces.  It is also an underlying  reason for appearance of unavoidable reducibility relations in 2D gauge theories.


In order to stress the role of the subring  $\bar{\mathcal{R}}\subset \mathcal{R}$ we will interpret the ring $\mathcal{R}$ as
$\mathcal{R}=\bar{\mathcal{R}}[\partial]$ and write its general element in the form
$$
  a=a_k\partial^{k}+a_{k-1}\partial^{k-1}+\cdots+ a_0\,,\qquad a_i\in \bar{\mathcal{R}}\,.
$$
If $a_k\neq 0$ we say that $a$ has \textit{degree} $k$ and write ${\mathrm{deg}}\,a=k$. For the sake of uniformity we also put $\deg\, 0=-\infty$. Then $\deg (ab)=\deg\, a+\deg\, b$. Clearly, the elements of degree $\leq k$ form a free $\bar{\mathcal{R}}$-module with the basis $1,\partial,\ldots,\partial^k$. Denoting this $\bar{\mathcal{R}}$-module by ${\mathcal{R}}_k$, we define the increasing filtration of $\mathcal{R}$:
$$
  \bar{\mathcal{R}}={\mathcal{R}}_0\subset {\mathcal{R}}_1\subset\cdots\subset \mathcal{R}\,,\qquad \mathcal{R}_k\cdot \mathcal{R}_l=\mathcal{R}_{k+l}\,.
$$

The notion of degree can be naturally extended from $\mathcal{R}$ to the $\mathcal{R}$-module $\mathcal{R}^{n}$ and then to any of its submodules. Namely, for any  $M=(M^i)\in \mathcal{R}^{n}$ we set
$$
\mathrm{deg}\, M= \max_{i} \mathrm{deg}\, M^i\,.
$$
Now any submodule $\mathcal{M}\subset \mathcal{R}^n$ can be endowed with an increasing filtration
$$
\mathcal{M}_0\subset \mathcal{M}_1\subset \mathcal{M}_2\subset \cdots\subset \mathcal{M}\,,\qquad \mathcal{R}_k\cdot \mathcal{M}_l\subset \mathcal{M}_{k+l}\,,
$$
where the $\bar{\mathcal{R}}$-module $\mathcal{M}_l$ is constituted by the elements of $\mathcal{M}$ of degree $\leq l$. It is important that each $\mathcal{M}_l$ admits a finite bases.

\section{Gauge symmetries and reducibility relations}\label{GS}

The key point of our approach to construction of gauge symmetry generators is an embedding of the original dynamics into a variational one. So, we start this section with description  of this embedding and discussion of its properties.

\subsection{The Pontryagin action} Given the field equations (\ref{RNF}), we introduce the new fields $\pi_i$,  $\mu_J$. The dynamics of the extended set of fields $\psi=\{\phi,\lambda, \pi, \mu\}$ are governed by the action functional
\begin{equation}\label{S}
  S[\phi,\lambda,\pi,\mu] =\int dxd\bar x\left(\pi_iT^i+\mu_J\overline{T}{}^J\right)=\int dxd\bar x\left(\pi_i\partial \phi^i - \mathcal{H} \right)\,,
\end{equation}
with
\begin{equation}\label{H}
\mathcal{H}=\pi_iZ^i-\mu_J\left( \bar \partial \phi^J-Z^J\right)\,.
\end{equation}
The functional (\ref{S}) has the form of Hamiltonian action if one treats  $x$ as time and the fields  $\phi^i$ and $\pi_i$ as the pairs of canonically conjugate variables, namely,
$$
\{\phi^i(\bar x), \pi_{j}(\bar x{}')\}=\delta^i_{j}\delta(\bar x-\bar x{}')
$$
and the Poisson brackets of all other variables vanish. In addition to $\phi$'s and $\pi$'s the Hamiltonian density (\ref{H}) depends on the fields $\lambda^\alpha$ and $\mu_J$. The equations of motion resulting from variation of (\ref{S}) read
\begin{equation}\label{HE}
  \begin{array}{ll}
  T^i=\displaystyle\frac{\delta S}{\delta \pi_i}= \partial \phi^i-\{\phi^i, H\}=0\,,\qquad & \displaystyle\displaystyle T_i= \frac{\delta S}{\delta \phi^i}= -\partial \pi_i+\{\pi_i, H\}=0\,,\\[5mm]
  \displaystyle \overline{T}{}^J=\frac{\delta S}{\delta \mu_J}=\bar\partial \phi^J-Z^J=0\,,\qquad&\displaystyle
  T_\alpha=\frac{\delta S}{\delta \lambda^\alpha}= -\sum_{n=0}^p(-\bar\partial)^n\left( Z^i_{\alpha n}\pi_i\right)=0\,,
  \end{array}
\end{equation}
where
\begin{equation}
H=\int d\bar x \,\mathcal{H}\,,\qquad Z_{\alpha n}^i=\frac{\partial Z^i}{\partial(\bar\partial^n \lambda^\alpha)}\,.
\end{equation}
 By construction, the equations in the first column (\ref{HE}) coincide with the original equations of motion (\ref{RNF}). The equations in the second column define the time evolution of the momenta.
 Variation by the fields $\mu_J$ and  $\lambda^\alpha$ yields the equations in the second line  (\ref{HE}). These, being independent of  the time derivatives of fields, can be regarded as (differential) constraints on the initial values of $\phi^i$, $\pi_j$, and $\lambda^\alpha$.

 In such a way the dynamics of the fields $\phi$'s and $\lambda$'s appear to be embedded into the wider  dynamics governed by the least action principle. Unfortunately, there is no commonly accepted name for this embedding.  In the mathematical theory of optimal control \cite{AS}, the action (\ref{S}) was  introduced by Lev Pontryagin in connection with his famous ``Pontryagin Minimum Principle''. For this reason we  will refer to (\ref{S}) as the \textit{Pontryagin action}. A characteristic feature of the Pontryagin action is its linear dependence of momenta.  In optimal control theory the fields $\lambda^\alpha$ and $\phi^i$ are referred to as the control and the state variables, respectively, and {minimization} of the functional (\ref{S}) solves the time-optimal problem.

 Yet another interpretation of (\ref{S}) is possible if one  starts from the equations in the ``strong'' normal form (\ref{NF}). In that case both $\mu_J$ and $\lambda^\alpha$ enter the Hamiltonian (\ref{H}) linearly, that is, as the Lagrange multipliers to the (primary) Hamiltonian constraints on $\phi$'s and $\pi$'s, and the Pontryagin action (\ref{S}) describes constrained Hamiltonian dynamics in Dirac's sense \cite{Dirac}, \cite{HT}.

 By an abuse of terminology, we will refer to the equations in the second line of (\ref{HE}) as Hamiltonian constraints, even though the functions $T_\alpha$ depend on the ``Lagrange multipliers'' $\lambda^\alpha$.\footnote{This terminology is justified by considering the equations $T_\alpha=0$ as linear constraints on the momenta $\pi$'s, not on $\phi$'s or $\lambda$'s. Then the same treatment applies to all the secondary constraints. Linearity in momenta ensures the existence of a solution, e.g. the constraints are satisfied by  $\pi_i=0$. It is the solutions with $\pi_i=0$ and $\mu_J=0$ that are naturally identified with the solutions to the original equations (\ref{RNF}).}

Now we claim that any gauge symmetry of the original equations (\ref{RNF}) gives rise to a gauge symmetry of the Pontryagin action (\ref{S}).  Indeed, in view of regularity of the field  equations (\ref{RNF}) the weak equalities (\ref{deltaT})  can be written as the ``strong'' ones
\begin{equation}\label{}
  \delta_\varepsilon T^i= {A}^i_j T^j+ {B}^i_J \overline{T}{}^J\,,\qquad \delta_\varepsilon \overline{T}{}^J= {C}^J_iT^i + {D}^J_I\overline{T}{}^I
\end{equation}
for some matrix differential operators ${A}$, ${B}$, ${C}$, and ${D}$ with coefficients depending on $\phi$'s, $\lambda$'s, $\varepsilon$ and their derivatives. More explicitly,
$$
  A^i_j=\sum_{q,p,n,m}(\partial^q\bar\partial{}^p\varepsilon) A_{qpnm}{}^i_j\partial^n\bar\partial{}^m\,, \qquad A_{qpnm}{}^i_j\in \mathfrak{F}\,,
$$
and the similar structure is assumed for the other operators $B$, $C$, and $D$. Let us define the formal transpose of a matrix differential operator ${U}=({U}_\alpha^A)\in \mathfrak{R}^{n\times m}$ as a unique  operator ${U}^\ast=({U}^{\ast A}_\alpha)\in \mathfrak{R}^{m\times n}$ satisfying the condition
\begin{equation}\label{dual}
  \int dxd\bar x \left(\eta_A {U}_\alpha^A\xi^\alpha\right)= \int dxd\bar x\left( \xi^\alpha {U}^{\ast A}_\alpha \eta_A \right)
\end{equation}
for any  compactly supported functions $\eta$'s and $\xi$'s.  With the definitions above we can extend the gauge symmetry (\ref{GT}) of the equations (\ref{RNF}) to that of the action (\ref{S}) by setting
\begin{equation}\label{GT-EXT}
  \delta_\varepsilon \pi_i=-{A}^{\ast j}_i\pi_j-{C}_i^{\ast J}\mu_J\,,\qquad \delta_{\varepsilon} \mu_J=-{B}^{\ast i}_J\pi_i- {D}^{\ast I}_J\mu_I\,.
\end{equation}
Clearly,   $\delta_\varepsilon S=0$ for any  $\varepsilon$ with compact support. The converse is also true: any gauge transformation of the form (\ref{GT}),  (\ref{GT-EXT}) that leaves invariant the action (\ref{S}) defines a gauge symmetry of the original equations (\ref{RNF}).

By Noether's second theorem, there is a one-to-one correspondence between the gauge symmetries of an action functional and the differential identities for the corresponding equations of motion.  Indeed, if
 \begin{equation}\label{GTS}
 \delta_\varepsilon \psi^A={R}_\alpha^A\varepsilon^\alpha
 \end{equation}
 is a family of gauge symmetries of the action (\ref{S}), then by definition
 \begin{equation}\label{GI}
 {R}_\alpha^{\ast A}\frac{\delta S}{\delta \psi^A}=0\,,
 \end{equation}
   and the last equality is nothing else but the Noether identity for the constrained Hamiltonian equations (\ref{HE}). Conversely, given the identity (\ref{GI}), the variation (\ref{GTS}) defines a gauge symmetry transformation with $R=(R^\ast)^\ast$.

  Thus, the problem of finding gauge transformations for the field equation (\ref{RNF}) can be replaced by that of finding the Noether identities for the  extended system of equations (\ref{HE}). Some of these identities come from the gauge symmetries of the original field equations, but the other do not. The definition of the Cartan normal form implies the compatibility condition (\ref{CC}) to hold. Regarding this condition as a set of the Noether identities for the equations (\ref{HE}), we get the following gauge transformations of the Pontryagin action:
  \begin{equation}\label{NGT}
  \delta_\varepsilon \mu_J=\partial \varepsilon_J+ U_J^{\ast I}\varepsilon_I\,,\qquad\delta_\varepsilon \pi_i=V_i^{\ast J}\varepsilon_J\,,\qquad \delta_{\varepsilon}\phi^i=0\,,\qquad \delta_{\varepsilon}\lambda^\alpha=0\,.
  \end{equation}
These transformations, having no effect on the original fields, do not correspond to gauge symmetries of the original theory.  As we will see in a moment, relations (\ref{GT}),(\ref{GT-EXT}) and (\ref{NGT}) exhaust all the gauge symmetries of the Pontryagin action.

 \subsection{The Noether identities} Notice that the Hamiltonian equations in the first  line (\ref{HE}) are solved for the time derivatives of the phase-space variables  $\phi^i$ and $\pi_j$. Therefore these equations are linear independent by themselves.  At the same time the Hamiltonian constraints in (\ref{HE}) contain no time derivatives of the phase-space variables at all. This implies that any Noether identity for (\ref{HE}) must necessarily involve time derivatives of the Hamiltonian constraints. According to (\ref{CC}) the time derivatives of the constraints $\overline{T}{}^J$ reduce immediately to linear combinations of the original equations (\ref{RNF}) and this yields the gauge transformations (\ref{NGT}). So, when looking for the gauge transformations of the original field equations (\ref{RNF}), one can disregard the Noether identities involving the time derivatives of $\overline{T}{}^J$. The time derivatives of the remaining Hamiltonian constraints are given by
\begin{equation}\label{dpat}
\partial T_\alpha \approx \sum_n\frac{\partial T_\alpha}{\partial(\bar\partial^n\lambda_\beta)}\partial\bar\partial^n\lambda^\beta+\{T_\alpha, H\}\,.
\end{equation}
As above, the sign of weak equality means ``modulo equations of motion'' (\ref{HE}).
To write down the general expression for the $k$-th time derivative we introduce the operator
$$
D=\partial_\lambda-\{T,\;\cdot\; \}\,, \qquad T=\int d\bar x \pi_i Z^i\,,
$$
where $\partial_\lambda$ denotes  the action of $\partial$ on $\lambda$'s.
By induction on $k$, one can see that
\begin{equation}\label{DDT}
\partial^k T_\alpha \approx D^k T_\alpha\,.
\end{equation}
Here we took into account that
\begin{equation}\label{TT}
\{T_\alpha (\bar x),\overline{T}^J(\bar x{}')\}\approx 0
\end{equation}
as a consequence of the compatibility condition (\ref{CC}). This allows us to omit the constraints $\overline{T}^J$ in the Hamiltonian and replace $H$ by $T$ in (\ref{dpat}).

Notice that the expression in the right hand side of (\ref{DDT}) does not involve the time derivatives of $\phi$'s and $\pi$'s, though it depends on the time derivatives of $\lambda$'s. Another peculiar property of the functions $D^k T_\alpha$ is their linear dependence of $\pi$'s. This is a simple consequence of the fact that the canonical Poisson brackets of linear in momenta function(al)s are again linear in momenta.  So, we have
\begin{equation}\label{DnT}
  D^k T_\alpha=\sum_{m=0}^M T_{\alpha m}^{ik}\bar\partial^m \pi_i\,,
\end{equation}
   where the coefficients $T^{ik}_{\alpha m}$ are real-meromorphic functions of finite number of arguments $\phi^i$,  $\bar \partial^k \phi^a$, and $\partial^p\bar\partial^q\lambda^\alpha$, that is,  $T^{ik}_{\alpha m}\in \mathcal{F}$.

Thus, we are  led to the conclusion that the existence of Noether's identities for the Hamiltonian equations (\ref{HE}) amounts  to the existence of identities for the successive time derivatives (\ref{DDT}), (\ref{DnT}) of the Hamiltonian constraints $T_\alpha$.  The latter identities have the form
\begin{equation}\label{BDT}
\sum_{k=0}^K {M}_k^\alpha D^k T_\alpha \approx 0 \,,\qquad M_k^\alpha\in \bar{\mathcal{R}}\,.
\end{equation}
A ``strong'' form of  the weak equalities (\ref{DDT}) and (\ref{BDT}) is
\begin{equation}\label{SF}
D^kT_\alpha=\partial^kT_\alpha + {K}^{ki}_\alpha T_i + {L}^k_{\alpha i}T^i+Q^k_{\alpha J} {\overline{T}}^J\,,\qquad \sum_{k=0}^K {M}_k^\alpha D^k T_\alpha =P_J \overline{T}{}^J\,.
\end{equation}
Here ${K}$, ${L}$, $Q$, and $P$ are given by some differential operators in $x$ and $\bar x$ with coefficients depending on $\phi$, $\lambda$, $\pi$, $\mu$, and their derivatives.   Combining relations (\ref{SF}), we get the  Noether identity
\begin{equation}\label{NI}
\sum_{k=0}^K {M}^\alpha_k\partial^k T_\alpha =  ({P}_J - {M}^\alpha_k Q_{\alpha J}^k )\overline{T}{}^J - {M}^\alpha_k {K}^{ki}_\alpha T_i-{M}^\alpha_k {L}^k_{\alpha i}T^i\,.
\end{equation}
In accordance with (\ref{GTS}) and (\ref{GI})  the corresponding gauge transformation reads
\begin{equation}\label{GaugeTr}
\begin{array}{ll}
\delta_{\varepsilon}\phi^i={K}^{\ast k i}_\alpha {M}^{\ast \alpha}_k\varepsilon\,,\qquad&\displaystyle  \delta_{\varepsilon} \lambda^\alpha = \sum_{k=0}^K (-\partial)^k {M}^{\ast \alpha}_k \varepsilon\,,\\[5mm]
\delta_{\varepsilon}\pi_i={L}_{\alpha i}^{\ast k}{M}_k^{\ast \alpha}\varepsilon \,,\qquad & \delta_\varepsilon \mu_J=  (Q_{\alpha J}^{\ast k}{M}^{\ast \alpha}_k -P^\ast_J)\varepsilon\,.
\end{array}
\end{equation}
As we have explained above the transformations in the first line define a gauge symmetry of the original equations (\ref{RNF}).

\subsection{The resolution}\label{res} The challenge now is to construct a complete set of identities (\ref{BDT}). This will require some algebraic technique.

Let us think of the momenta $\pi_i$ as a bases of the left $\bar{\mathcal{R}}$-module $\bar{\mathcal{R}}^n$. Then the Hamiltonian constraints $T_\alpha$, being linear in $\pi$'s, generate a submodule of $\bar{\mathcal{R}}^n$, which we denote by $\mathcal{T}_0$. Starting from $\mathcal{T}_0$ we define by induction the sequence of left $\bar{\mathcal{R}}$-modules
\begin{equation}\label{Tf}
  \mathcal{T}_{k+1}=\mathcal{T}_k \cup  D \mathcal{T}_k\,.
\end{equation}
According to this definition
\begin{equation}\label{}
  \mathcal{T}_k=\mathrm{span}_{\bar{\mathcal{R}}}\big\{T_\alpha, D T_{\alpha}, \ldots , D^k T_{\alpha}\big\}\,.
\end{equation}
  By Theorem \ref{Th}, each module $\mathcal{T}_k$ is free and admits a finite basis. It is clear that $$\mathrm{rank} \,\mathcal{T}_k\leq \mathrm{rank}\, \mathcal{T}_{k+1}\leq n\,.$$
Since the module $\bar{\mathcal{R}}^n$ is Noetherian, the ascending chain of its submodules
\begin{equation}\label{asc}
  \mathcal{T}_0\subset \mathcal{T}_1\subset \mathcal{T}_2\subset\cdots \subset \bar{\mathcal{R}}^n
\end{equation}
eventually stabilizes, i.e., there exists an integer  $K$ such that $\mathcal{T}_{K}=\mathcal{T}_{K+1}\equiv \mathcal{T}$.

Formula (\ref{DDT}) defines the natural action  of the operator $\partial$ in $\mathcal{T}$, so that we  can think of $\mathcal{T}$ as a differential module, or still better, as the left $\mathcal{R}$-module with the generating set  ${\mathbf{T}}=\{T_\alpha\}$.
Let $\mathcal{M}=\mathrm{Syz}({\mathbf{T}})$ denote the corresponding syzygy module. By definition, $\mathcal{M}$ consists of covectors $M\in \mathcal{R}^l$
satisfying the condition
$$
M^\alpha T_\alpha=0\,.
$$
It is the syzygy module $\mathcal{M}$ that describes all the linear relations between the constraints (\ref{BDT}). The module $\mathcal{M}$ enjoys the increasing filtration
$$
\mathcal{M}_0\subset \mathcal{M}_1\subset \mathcal{M}_2\subset \cdots \subset \mathcal{M}\,,\qquad \mathcal{R}_l\cdot \mathcal{M}_k\subset \mathcal{M}_{k+l}\,,
$$
where the left $\bar{\mathcal{R}}$-module $\mathcal{M}_k$ consists of the syzygies of degree $\leq k$.

Stabilization of the sequence (\ref{asc}) at $K$-th term  implies that
$$
\partial^{K+1}T_\alpha = {M}_{K \alpha}^\beta \partial^{K}T_\beta+{M}_{K-1, \alpha}^\beta \partial^{K-1}T_\beta+\cdots +{M}_{0 \alpha}^\beta T_\beta
$$
for some ${M}^\alpha_{k \beta}\in \bar{\mathcal{R}}$. In other words, the covectors $M_\alpha\in \mathcal{R}^l$ with components
\begin{equation}\label{M}
M_\alpha^\beta=\delta_\alpha^\beta\partial^{K+1}-{M}_{K \alpha}^\beta \partial^{K}-{M}_{K-1, \alpha}^\beta \partial^{K-1}-\cdots -{M}_{0 \alpha}^\beta
\end{equation}
define a set of $l$ syzygies of degree $K+1$.

\begin{prop}
For any $N=(N^\alpha)\in {\mathcal{R}}^l$ we have $\deg(N^\alpha M_\alpha)=K+1 + \deg N$.
\end{prop}
The proof is straightforward  and we leave it to the reader.
As an immediate consequence of the proposition above we have

\begin{cor}\label{cor}
The covectors $M_\alpha$ are linearly independent over $\mathcal{R}$. Furthermore, there are no nontrivial linear combinations of $M_\alpha$ of degree $< K+1$, that is, $N^\alpha M_\alpha \in \mathcal{M}_K$ implies $N = 0$.
\end{cor}

Of course, the $\bar{\mathcal{R}}$-module $\mathcal{M}_{K+1}$ may also contain elements of degree $\leq K$, which constitute the module $\mathcal{M}_{K}$. Again, by Theorem \ref{Th} the module $\mathcal{M}_K$ is free and finitely generated. Writing $\{M_A\}$ for a  basis in $\mathcal{M}_K$, we set ${\mathbf{M}}=\{M_\alpha\}\cup\{ M_A\}$ and denote by $r_1$ the number of elements of ${\mathbf{M}}$. By definition, $r_1\geq l$.

\begin{prop}
${\mathbf{M}}$ is a generating set of the $\mathcal{R}$-module $\mathcal{M}$ and a basis of the $\bar{\mathcal{R}}$-module $\mathcal{M}_{K+1}$.
\end{prop}

\begin{proof}
The covectors $M\in \mathcal{M}$ of degree $\leq K$ belong to the $\mathrm{span}_{\bar{\mathcal{R}}}({\mathbf{M}})$ by the definition of ${\mathbf{M}}$. If $\mathrm{ord}\, M=L\geq K+1$, then
$$
M^\alpha=\bar B^\alpha \partial^{L}+\cdots\,,
$$
where $\bar B^\alpha\in \bar{\mathcal{R}}$ and dots stand for the terms of degree $<L$.
It is clear that the order of the covector
$$
M'=M-\bar B^\alpha\partial^{L-K-1}M_\alpha \in \mathcal{M}
$$
is less than $L$. Proceeding in this way, we obtain  a set of coefficients  $B^\alpha\in \mathcal{R}$ such that $M-B^\alpha M_\alpha\in \mathcal{M}_K$. Expanding the last difference by $M_A$, we express $M$ as a linear combination of the generators from $\mathbf{M}$.

The consideration above shows also that ${\mathbf{M}}$ generates $\mathcal{M}_{K+1}$. So, it remains to prove that the elements of ${\mathbf{M}}$ are linearly independent over $\bar{\mathcal{R}}$. By Corollary \ref{cor}, there are no nontrivial relations between $M_\alpha$ and $M_A$. On the other hand, the covectors of both the groups $\{M_\alpha\}$ and $\{M_A\}$ are obviously linearly independent among themselves.

\end{proof}

Given the generating set ${\mathbf{M}}$, we can define its syzygy module $\mathcal{N}=\mathrm{Syz}({\mathbf{M}})$. The main property of the $\mathcal{R}$-module $\mathcal{N}$ is described by the following

\begin{prop} \label{N-basis}
$\mathcal{N}$ is a free module of rank $r_2=r_1-l$.
\end{prop}

\begin{proof}
By definition, $\partial M_A\in \mathcal{M}_{K+1}$. Using the basis $\mathbf{M}=\{M_\alpha\}\cup\{ M_A\}$ in $\mathcal{M}_{K+1}$, we can write
\begin{equation}\label{N-syz}
\partial M_A =K_A^\alpha M_\alpha+L_A^B M_B
\end{equation}
for some $K_A^\alpha, L_A^B\in \bar{\mathcal{R}}$. It is clear that the coefficients $K$ and $L$ are uniquely defined by (\ref{N-syz}) and $K_A^\alpha = 0$ whenever $\mathrm{ord}\,M_A <K$. Relation (\ref{N-syz}) defines the set $\mathbf{N}=\{N_A\}$ of $m$ syzygies of $\mathbf{M}$,
$$
N_A =(-K_A^\alpha, \delta_A^B\partial-L_A^B)\,,\qquad \deg\, N_A=1\,.
$$

We claim that the syzygies $N_A$ are linear independent over $\mathcal{R}$ and generate the module $\mathcal{N}$. Let $N=(N^\alpha, N^A)$ be a syzygy of $\mathbf{M}$, so that
$
N^\alpha M_\alpha +N^AM_A=0
$.
Subtracting from $N$ an appropriate linear combination of (\ref{N-syz}), we can always decrees the degree of the components $N^A$ to zero. In other words, there are  $B^A\in \mathcal{R}$ such that
$$
\tilde{N}=N - B^AN_A =(\tilde{N}^\alpha,  \tilde{N}^A)\in \mathcal{N}\,,
$$
where $\tilde{N}^\alpha\in \mathcal{R}$ and ${\tilde{N}}^A\in \bar{\mathcal{R}}$.
By Corollary \ref{cor}, the equality $\tilde{N}^AM_A+\tilde{N}^\alpha M_\alpha=0$ implies
 $\tilde{N}=0$, and  hence $N=B^AN_A$. This proves that $\mathbf{N}$ is a generating set of the syzygy module $\mathcal{N}$.

To prove the linear independence of $N_A$, we simply observe that $B^AN_A=0 $ implies  $K^A=B^A\partial-B^CL_C^A=0$. It is clear that $\deg\, K^A=\deg\, B^A+1$ and $K^A=0$ implies  $B^A=0$.
\end{proof}

To formulate the next proposition we recall the notion of a dual $\mathcal{R}$-module. It is obtained by abstracting the dualization procedure for differential operators (\ref{dual}) that results  from ``integration by parts''. First we define the dual of the ring $\mathcal{R}$ itself. By definition, the dualization map $\mathcal{R}\rightarrow \mathcal{R}^\ast$ is an isomorphism of the underlying $\mathbb{R}$-vector space such that
\begin{equation}\label{d-prop}
(f
\cdot \partial^n\bar\partial^m)^\ast=(-1)^{n+m}\partial^n\bar\partial^m \cdot f\qquad \forall f\in \mathcal{F}\subset \mathcal{R}\,.
\end{equation}
(The dot stands for multiplication in $\mathcal{R}$.)
It follows from the definition that
\begin{equation}\label{ab}
(a\cdot b)^\ast=b^\ast\cdot a^\ast \,,\qquad (a^\ast)^\ast=a \qquad \forall a, b\in \mathcal{R}\,.\end{equation}
Setting $(\mathcal{R}^n)^\ast=(\mathcal{R}^\ast)^n$, we extend the dualization map to the free, finitely generated $\mathcal{R}$-modules and all their submodules. As is seen from (\ref{ab}) dualization revers the order of multipliers,  turning  left $\mathcal{R}$-modules into right $\mathcal{R}$-modules, and vice versa. Furthermore, the dualization map is involutive, that is,  $(\mathcal{M}^\ast)^\ast=\mathcal{M}$ for any submodule $\mathcal{M}\subset \mathcal{R}^n$. Using the last property, one can see that the module $\mathcal{M}$ is free iff the dual module $\mathcal{M}^\ast$ is free.

Due  to the Noether correspondence  between the  gauge transformations (\ref{GTS}) and the identities (\ref{GI}) we have the following
 \begin{prop}\label{G=M}
Formulae (\ref{GaugeTr}) establish an isomorphism between the right $\mathcal{R}$-module of the gauge symmetries $\mathcal{G}$ and the dual of the left $\mathcal{R}$-module $\mathcal{M}=\mathrm{Syz}({\mathbf{T}})$, i.e., $\mathcal{G}\simeq \mathcal{M}^\ast$.
\end{prop}

If we introduce the collective notation $M_I$ for the covectors from the generating set $\mathbf{M}=\{M_\alpha\}\cup\{M_A\}$ of $\mathcal{M}$, then the corresponding generators (\ref{R}) of $\mathcal{G}$ can be written as
\begin{equation}\label{RR}
R_I=\left(
           \begin{array}{c}
             R_I^i \\
             R_I^\alpha \\
           \end{array}
         \right) =\left(
                                \begin{array}{c}
                                  K_\alpha ^{\ast ki} M^{\ast\alpha}_{kI} \\
                                  M_I^{\ast \alpha} \\
                                \end{array}
                              \right)\,,
\end{equation}
where
$$
M_I^\alpha = \sum_{k=0}^{K+1}M_{Ik}^\alpha\partial^k\,,\qquad M_{Ik}^\alpha\in \bar{\mathcal{R}}\,,
$$
and the coefficients $K^{ki}_\alpha$ are defined by relation (\ref{SF}).

Given the generating set $\mathbf{R}=\{R_I\}$ of $\mathcal{G}$,  the isomorphism stated by Proposition \ref{G=M} implies that $\mathcal{Z}=\mathrm{Syz}(\mathbf{R})\simeq \mathcal{N}^\ast$, where   $\mathcal{N}=\mathrm{Syz}(\mathbf{M})$. The module $\mathcal{Z}$, being dual to the free module $\mathcal{N}$,  is free and we can take the vectors $\mathbf{Z}=\{N_A^\ast\}$ as its basis. By construction, $R_I N_A^{\ast I}=0$. Now, arranging the vectors of the generating sets $\mathbf{R}$ and $\mathbf{Z}$ into rectangular matrices and treating these matrices  as homomorphisms of free, right $\mathcal{R}$-modules, we get the desired free resolution (\ref{MR}). The ranks of the free modules in (\ref{MR}) satisfy the relation $r_1-r_2=l$, which, being translated in physical language, just says that the number of ``independent gauge symmetries'' is equal to $l$. This is consistent with the heuristic idea that the number of independent gauge parameters should coincide with the number of arbitrary functional parameters - the fields $\lambda^\alpha$ - in the general solution to the field equations.

\section{An example}

By way of illustration we will consider the following equation:
\begin{equation}\label{dA}
\partial_\mu A^\mu +\frac g2 A_\mu A^\mu=0\,.
\end{equation}
Here,  $A^\mu$ is a vector field in two-dimensional Minkowski space and the index $\mu=0,1$ is raised and lowered with the help of the Minkowski metric.

In the special case of $g=0$, the equation enjoys the irreducible  gauge symmetry
\begin{equation}\label{de}
\delta_\varrho A^\mu =\epsilon^{\mu\nu}\partial_\nu\varrho\,,
\end{equation}
where $\epsilon^{\mu\nu}=-\epsilon^{\nu\mu}$ is the Levi-Civita symbol. This symmetry is an immediate consequence of the Poincar\'e Lemma for the differential forms on plane. A simple count shows that the gauge transformation  (\ref{de}) leaves no room for the physical degrees of freedom, so that the theory appears to be  topological at the free level. For $g\neq 0$, the model (\ref{dA}) keeps to be  topological, while the structure of gauge symmetry becomes much more complicated.

In order to make contact with the notation of the previous sections we set
$$
x^0=x\,,\qquad x^1=\bar x\,,\qquad A^0=-A_0=\phi\,,\qquad A^1=A_1=\lambda\,.
$$
Then, the field equation (\ref{dA}) takes the Cartan normal form
\begin{equation}\label{df}
\partial \phi+\bar\partial \lambda-\frac g2\phi^2+\frac g2 \lambda^2=0\,.
\end{equation}
The corresponding Pontryagin action reads
$$
S=\int \pi\Big( \partial\phi +\bar\partial\lambda-\frac g2\phi^2+\frac g2\lambda^2 \Big)dxd\bar x\,.
$$
Varying this action, we get (\ref{df}) plus the pair of equations
\begin{equation}\label{TP}
T=\frac{\delta S}{\delta \lambda}=-\bar\partial\pi +g\lambda \pi=0\,,\qquad \Pi=\frac{\delta S}{\delta \phi}=-\partial\pi-g\phi\pi=0\,.
\end{equation}
The time derivative of the Hamiltonian constraint $T$ is given by
\begin{equation}\label{TPF}
\partial T =(\bar\partial-g\lambda)\Pi-g\phi T+g F\pi\,,
\end{equation}
where
$
F\equiv \bar\partial \phi+\partial \lambda
$.
In case $g=0$, relation (\ref{TPF}) turns to the identity
$$
\partial T -\bar\partial \Pi=0
$$
for the equations (\ref{TP}). In line with Proposition \ref{G=M}, this identity gives rise to the gauge transformation (\ref{de}).

If $g\neq 0$, then the first derivative $\partial T$ is not proportional to $T$ and the other equations, and we need to know the second time derivative of $T$.
Differentiating, we find
\begin{equation}\label{d2T}
\partial^2T=\partial(\bar\partial -g\lambda)\Pi -gF\Pi -g \partial (\phi T)+g\pi (\partial -g \phi)F\,.
\end{equation}
Solving (\ref{TPF}) for $\pi$, we get
\begin{equation}\label{pi}
\pi=\frac{(\partial +g\phi)T-(\bar\partial -g\lambda)\Pi}{gF}\,.
\end{equation}
Substituting  (\ref{pi}) into (\ref{d2T}) yields the identity for the equations (\ref{TP})
\begin{equation}\label{d2t-id}
\partial^2T=\partial(\bar\partial -g\lambda)\Pi-gF\Pi -g\partial(\phi T) + \frac{(\partial +g\phi)T-(\bar\partial -g\lambda)\Pi}{F}(\partial -g\phi)F\,,
\end{equation}
which is equivalent to
\begin{equation}\label{id1}
(\partial + g\phi)\big(F^{-1}(\partial +g\phi)T\big)-(\partial +g\phi)\big(F^{-1}(\bar\partial -g\lambda)\Pi\big)-g\Pi=0\,.
\end{equation}
The corresponding gauge transformation reads
\begin{equation}\label{gt1}
\delta_{\varepsilon'}\phi=g\varepsilon'+(\bar\partial +g\lambda)\big(F^{-1}(-\partial +g\phi)\varepsilon'\big)\,,\qquad \delta_{\varepsilon'} \lambda =(-\partial +g\phi)\big(F^{-1}(-\partial +g\phi)\varepsilon'\big)\,.
\end{equation}

 The identity (\ref{d2t-id}) also says that the left $\mathcal{R}$-module $\mathcal{T}$ defined in Sec. \ref{res} is generated over $\bar{\mathcal{R}}$ by $T$ and $\partial T$. Since the rank of $\mathcal{T}$  is equal to $1$, the generators $T$ and $\partial T$ are linearly dependent over $\bar{\mathcal{R}}$. To obtain a linear relation between them we just apply  the operator $g(\bar\partial -g\lambda)$ to both sides of (\ref{pi}). This gives one more identity for the equations (\ref{TP})
\begin{equation}\label{id2}
(\bar\partial -g\lambda)\big(F^{-1}(\partial + g\phi)T\big)+gT-(\bar\partial -g\lambda)\big(F^{-1}(\bar\partial - g\lambda)\Pi\big)=0
\end{equation}
and one more gauge transformation for the original fields
\begin{equation}\label{gt2}
\delta_{\varepsilon''} \phi =-(\bar\partial +g\lambda)\big(F^{-1}(\bar\partial +g\lambda)\varepsilon''\big)\,,\qquad \delta_{\varepsilon''} \lambda = (\partial - g\phi)\big(F^{-1}(\bar\partial + g\lambda)\varepsilon'' \big) + g\varepsilon''\,.
\end{equation}

According to Proposition \ref{N-basis}, the number of ``independent gauge transformations'' coincides with the number of $\lambda$'s in the Cartan normal form of the field equations. In the case at hand, the latter is equal to $1$. Hence, the two-parameter gauge symmetry  (\ref{gt1}), (\ref{gt2}) is reducible. The generator of the corresponding gauge-for-gauge transformation can be read off from a syzygy for identities  (\ref{id1}) and (\ref{id2}). It is easy to check that multiplying (\ref{id2}) on the left by $\partial+g\phi$ and subtracting  the result from (\ref{id1}) multiplied  by $\bar \partial-g\lambda$, we get zero for any functions $T$ and $\Pi$. Upon dualization this syzygy gives the following  gauge-for-gauge transformation:
\begin{equation}\label{gfg}
\delta_\varepsilon \varepsilon'=(\bar\partial+g\lambda)\varepsilon\,,\qquad \delta_{\varepsilon}\varepsilon''=-(\partial-g\phi)\varepsilon\,.
\end{equation}

Both the gauge transformations and the reducibility relation can be rewritten in the Lorenz covariant form. To this end, we introduce the operator of ``covariant derivative''
$$
D_\mu=\partial_\mu+g A_\mu\,,
$$
with the curvature
$$
D_\mu\tilde{D}^\mu=gF\,,
$$
where
$$
F=\epsilon^{\mu\nu}\partial_\mu A_\nu\,, \qquad \tilde{D}^\mu=\epsilon^{\mu\nu}D_\nu\,.
$$
In terms of the covariant derivative, the variation of the field equation (\ref{dA}) is given by $D_\mu \delta A^\mu$ and the gauge transformations read\footnote{Here we restrict ourselves  to the general field configurations for  which $F\neq 0$.}
\begin{equation}\label{cgt}
\delta_\varepsilon A^\mu =g \varepsilon^\mu-\tilde{D}^\mu (F^{-1} D_\nu \varepsilon^\nu) \,,
\end{equation}
where $\varepsilon^\mu$ is an arbitrary vector parameter. These gauge transformations  reproduce (\ref{gt1}) and (\ref{gt2}) upon identification $(\varepsilon^0, \varepsilon^1)=(\varepsilon', \varepsilon'')$, while the  gauge-for-gauge transformation (\ref{gfg}) takes the form
\begin{equation}\label{crr}
\delta_\varepsilon \varepsilon^\mu=\tilde{D}^\mu\varepsilon\,,
\end{equation}
with $\varepsilon$ being an arbitrary scalar parameter.

To summarize, the field equation  (\ref{dA}) describes a gauge theory without physical degrees of freedom. For $g\neq 0$, the gauge symmetry is reducible and involves two gauge parameters with one gauge-for-gauge transformation. At $g=0$ the structure of gauge symmetry bifurcates: there is a  one-parameter gauge transformation, which is automatically irreducible.

What are the physical implications of this bifurcation phenomenon?
In our opinion it gives a striking counterexample to a widespread belief that any deformation of a free theory  by interaction can either \textit{deform} or \textit{break}  its gauge symmetry\footnote{In the physical literature, a simultaneous deformation of equations of motion and their gauge symmetries is known as the ``Noether procedure''.}:   It is not difficult to see that the gauge symmetry (\ref{de}) of the free field equation $\partial_\mu A^\mu=0$ admits \textit{no} deformation by $g$ to a gauge symmetry of the nonlinear equation (\ref{dA}), and yet the nonlinear equation is gauge invariant! To gain greater insight into how  this happens, it is instructive to look at the transformations (\ref{cgt}) and (\ref{crr}) in the ``abelian limit'' $g\rightarrow 0$. Setting the coupling constant to zero, we find
\begin{equation}\label{dAdE}
\delta_\varepsilon A^\mu=-\epsilon^{\mu\nu}\partial_\nu(F^{-1}\partial_\lambda\varepsilon^\lambda)\,,\qquad \delta_{\varepsilon}\varepsilon^\lambda =\epsilon^{\lambda\mu}\partial_\mu\varepsilon\,.
\end{equation}
As is seen, the gauge transformations for  $A_\mu$ reproduce (\ref{de}) with $\varrho=-F^{-1}\partial_\lambda \varepsilon^\lambda$.  Since the gauge parameters $\varepsilon^\mu$ enter these transformations through a single function $\varrho$, the gauge symmetry appears to be reducible. Taken together relations (\ref{dAdE}) allows one to gauge out as many degree of freedom as the single gauge transformation (\ref{de}).
So, one may regard (\ref{dAdE}) as just a queer form of the ``minimal'', i.e., irreducible,  gauge transformation (\ref{de}). It is this form, however, that admits a deformation to the gauge symmetry of the nonlinear theory (\ref{dA}).

Thus, the main lesson to learn is that a smooth deformation of field equations is not always
followed by a  smooth deformation of their gauge symmetry generators even when the number of physical degrees of freedom is preserved. The system may keep to be gauge invariant, possessing  the same number of physical degrees of freedom, while the structure of its gauge symmetry changes drastically. These results cast some doubt on the previous ``no-go'' theorems for the existence of consistent interactions in various field-theoretical models as all these theorems considered deformations of one particular set of gauge generators.  Similar to the example above, the minimal set of gauge generators may happen to resist any nontrivial deformation, while a non-minimal  and deformable set (if exists) is far from obvious.

\section{Overlook}

In this paper, we proposed a simple method for constructing a full
set of gauge symmetry generators for arbitrary system of field
equations in two dimensions. In the case of reducible gauge
generators, the method provides an explicit basis for reducibility
relations. At the heart of our construction is a special Hamiltonian
system - the Pontryagin system - that can be associated to any
system of PDEs in the Cartan normal form. Although the Pontraygin
system is not dynamically equivalent to the original one, it makes
possible to borrow some ideas and constructions from Dirac's
constrained dynamics.

The knowledge of the free resolution (\ref{MR}) allows one to
determine the number of physical degrees of freedom per point. This
number is perhaps  the most important physical characteristic of a
gauge system. The explicit  formula for this number can be found in
\cite{KLS}. When applied to our equations (\ref{RNF}), the formula
expresses the number of physical degrees of freedom in terms of the
number and the order of the differential operators defining the
gauge algebra generators $\mathbf{R}$ \textit{and} the generators
$\mathbf{Z}$ of reducibility relations. The formula makes no
difference between the space and the time coordinates being thus
invariant under the general coordinate transformations. On the other
hand, having brought a system of PDEs into the Cartan normal form,
one can try to take the advantage of separating the independent
variables into the time and the space coordinates. As we ague below,
the Cartan normal form suggests a more simple way for computing
physical degrees of freedom in 2D field theory: these can be
evaluated directly by a complete set of gauge symmetry generators
$\mathbf{R}$, without any knowledge of the reducibility relations
$\mathbf{Z}$.

In order to formulate our hypothesis we rewrite the general gauge transformation associated with  $\mathbf{R}$ in a form which sets apart the time derivatives of the gauge parameters from their space derivatives:
\begin{equation}\label{BR}
\delta_{\varepsilon} \phi^i=\sum_{k=1}^K\bar{R}^i_{Ik}\partial^k\varepsilon^I\,,\qquad \delta_\varepsilon\lambda^\alpha = \sum_{k=1}^{K+1}\tilde{R}^\alpha_{Ik}\partial^k\varepsilon^I\,.
\end{equation}
Here all $\bar{R}$'s belong to $\bar{\mathcal{R}}$. Let us introduce the multi-index $\Lambda =(I,k)$ that runs over $L$ different pairs and the matrix $\bar{\mathbf{R}}=(\bar{R}_B^i)\in \bar{\mathcal{R}}^{n\times L}$ determining the gauge transformations of $\phi$'s. Now we can formulate the following

\vspace{3mm}
\noindent
\textbf{Conjecture 1}. \textit{The number of physical degrees of freedom is equal to} $n-m-\mathrm{rank}\, \bar{\mathbf{R}}$.
\vspace{3mm}

Although  we are unable to prove this statement at the moment, we
can adduce some convincing arguments in its favor. Notice that the
maximal order of the time derivative of $\varepsilon$'s in the gauge
transformations of $\phi$'s is one less than the corresponding order
in the transformations of $\lambda$'s. This fact follows immediately
from the very structure of the field equations (\ref{RNF}).  The
time derivatives of the gauge parameters $\varepsilon^\Lambda \equiv
\partial^k\varepsilon^I$, $\Lambda=1,2,\ldots, L$, can be made
\textit{arbitrary} functions of $\bar{x}$ at each  given instant of
time $x=x_0$. By making use of these functions we can change
independently the values of at most $\max (L, n+l)$ components of
the fields $\phi$ and $\lambda$. The actual number of components
that can be affected by the gauge transformations (\ref{BR}) depends
on the structure of $\bar{R}$'s. As the infinitesimal gauge
transformations are linear in $\varepsilon^\Lambda$, one can
evaluate this number  by considering the rank of the corresponding
matrix. (Recall that the concept of rank is well defined for the
matrices over $\bar{\mathcal{R}}$, see Appendix A.) From formulae
(\ref{M}) and (\ref{RR}) one can see immediately  that the highest
time derivatives of the gauge parameters enter the gauge
transformations in such a way that it is possible to gauge out all
the components $\lambda^\alpha$  at the cost of the functions
$\partial^{K+1}\varepsilon^I$. So, one can regard $\lambda$ as a
purely gauge field\footnote{This is also agreed with the fact that
one can set $\lambda^\alpha$ to be totally arbitrary functions of
$x$ and $\bar x$ in the general solution to the field equations.}.
As to the remaining field $\phi$, the number of its  unaffected
components must be then given by the corank of the matrix
$\bar{\mathbf{R}}$. With allowance made for the $m$ constraints
$\overline{T}^J=0$ imposed  on the $n$ components of the field
$\phi$ this leaves  $n-m- \mathrm{rank} \,\bar{\mathbf{R}}$ physical
degrees of freedom per space point.

The most difficult part in a rigorous proof of the above conjecture is the independence of  $\mathrm{rank}\, \bar{\mathbf{R}}$ (and hence, the number of physical degrees of freedom) of the choice of a generating set $\mathbf{R}$ for the gauge symmetry transformations. In this situation, it is desirable to have a basis-independent definition for the value $\mathrm{rank}\, \bar{\mathbf{R}}$. For that end, consider the field equations in the strong normal form (\ref{NF}). The corresponding Pontryagin action (\ref{S}) describes a constrained Hamiltonian system with the primary Hamiltonian constraints  $T_\alpha$ and ${\overline{T}}^J$. To count the number of the physical degrees of freedom in this model one can follow the usual prescriptions of the constrained Hamiltonian formalism \cite{Dirac}, \cite{HT}. Namely, by applying the Dirac-Bergmann algorithm one first find all the secondary constraints ensuring the consistency of the whole dynamics. These constraints,
both primary and secondary, are then separated into first and second class. The number of physical degrees of freedom is obtained by subtracting from the dimension of the original phase space the number of the second class constraints and the doubled number of the first class constraints. It should be noted that the original system of field equations is not equivalent to that resulting from the Pontryagin action. The latter system contains extra degrees of freedom carried by the momenta $\pi_i$. Taking into account that the fields $\lambda^\alpha$ -
the Lagrange multipliers to the primary constraints $T_\alpha$ - are purely gauge and the fact that the physical degrees of freedom    are equally distributed between the ``position coordinates'' and momenta in the phase space, we arrive at the conclusion that the number of physical degrees of freedom  carried by the original fields is the half of the physical degrees of freedom of the  Pontryagin system.   These arguments can be further refined by making use of the special structure of the Hamiltonian constraints. According to (\ref{TT}) the holonomic constraints $\overline{T}^J$ are in involution with $T_\alpha$. So, they produce no secondary constraints. The secondary constraints result from the iterated Poisson brackets of the Hamiltonian constraints $T_\alpha$. Like the primary constraints $T_\alpha$, all their descendants  are linear in $\pi$'s. This makes possible to  regard them as the first class  constraints on the momenta.
The second class constraints do not appear in this theory.

As is well known in constrained dynamics \cite{Dirac}, \cite{HT}, the gauge symmetries are generated by the whole set of first class constraints through the Poisson bracket.  The Hamiltonian action of  the constraints
$\bar{T}^J$ do not affect the original  fields $\phi^i$ and can thus be regarded as producing a gauge transformation for the momenta $\pi_i$. At the same time, the equations $\bar{T}^J=0$ impose no restriction on the momenta $\pi_i$,
constraining exclusively the original fields $\phi^i$. Contrary to this, the constraints $T_\alpha=0$ can be regarded as linear differential equations for $\pi$'s (with  coefficients depending  on $\phi$'s). The Hamiltonian action of $T_\alpha$ transforms both
$\pi_i$  and $\phi^i$. Thus, restricting to the sector of original dynamics, we see that the space of fields $\phi$ is constrained  (at each point)
by the $m$ equations $\bar{T}^J=0$  and is foliated by the gauge orbits resulting from the Hamiltonian action of the
primary constraints $T_\alpha$ and all their descendants. Let $\{T_I\}$ denote the complete set of the Hamiltonian constraints (primary and secondary) not including the holonomic constraints $\overline{T}^J$. The constraints $T_I$, being by definition linear in momenta, can be interpreted as generators of a left $\bar{\mathcal{R}}$-module $\mathcal{D}\subset \bar{\mathcal{R}}^n$. By construction, the module $\mathcal{D}$ is closed with respect to the Poisson bracket. Geometrically, one can think  of $\mathcal{D}$ as an integrable distribution in the tangent bundle of the configuration space of fields $\phi$. The integral leaves of this distribution are then identified with the gauge orbits generated by the first-class constraints. As any submodule of $\bar{\mathcal{R}}^n$, the module $\mathcal{D}$ admits a finite basis. It gives a basis of the first-class constraints that are linear in momenta. Now subtracting from $n$ - the number of fields $\phi$ - the number of the holonomic constraints $m$ together with  the ``dimension'' of the gauge orbits, which is identified with $\mathrm{rank}\, \mathcal{D}$, we should obtain the number of physical degrees of freedom per point, that is,  $n-m-\mathrm{rank}\, \mathcal{D} $. Comparing this with Conjecture 1, we arrive at

\vspace{3mm}\noindent
\textbf{Conjecture 2}. $\mathrm{rank}\,\bar{\mathbf{R}}=\mathrm{rank}\, \mathcal{D}$.
\vspace{3mm}

Notice that the aforementioned possibility to choose a \textit{basis} in the set of first class constraints owns its existence in special algebraic properties of the ring of ordinary differential operators $\bar{\mathcal{R}}$ and, eventually, in two-dimensionality of the model. So, neither of the statements above applies directly to higher dimensional field theories.

Concluding this section, we would like to stress that our arguments in support of both the conjectures are mostly heuristic
and by no means substitute rigorous proofs.

\begin{appendix}
\section{Jacobson normal form}

 Given the ring $\bar{\mathcal{R}}=\mathcal{F}[\bar\partial]$, a square matrix $\mathbf{U}\in \bar{\mathcal{R}}^{n\times n}$ is called \textit{unimodular} if there exists a matrix $\mathbf{U}^{-1}\in \bar{\mathcal{R}}^{n\times n}$ such that $\mathbf{UU}^{-1}=\mathbf{U}^{-1}\mathbf{U}=1$. Recall the following fundamental result.

\vspace{3mm}
\noindent\textbf{Theorem} (on diagonal reduction). \textit{Given a matrix $\mathbf{M}\in \bar{\mathcal{R}}^{n\times m}$, there exist unimodular matrices $\mathbf{U}\in \bar{\mathcal{R}}^{n\times n}$ and   $\mathbf{V}\in \bar{\mathcal{R}}^{m\times m}$ such that
\begin{equation}\label{JNF}
  \mathbf{UMV}=\left(
        \begin{array}{cc}
          \mathbf{D} & 0 \\
          0 & 0 \\
        \end{array}
      \right)\,,
\end{equation}
where $\mathbf{D}=\mathrm{diag} (1,\ldots,1,\Delta)\in \bar{\mathcal{R}}^{l\times l}$ for some nonzero $\Delta\in \bar{\mathcal{R}}$.}

The number $l$ is called the \textit{rank of the matrix} $\mathbf{M}$. Notice that the decomposition (\ref{JNF}) is not unique, only the order of the differential operator $\Delta$ is uniquely defined.
 The right hand side of (\ref{JNF}) is known as the Jacobson normal form of the matrix $\mathbf{M}$ \footnote{Another name is the Teichm\"uller-Nakayama normal form.}. It can be viewed as a non-commutative generalization of the Smith decomposition for a matrix over a Euclidean domain.

Since the ring $\mathcal{R}$ is a principal ideal domain, the
matrices $\mathbf{U}$ and $\mathbf{V}$ can be obtained by performing elementary row and column operation. These include
\begin{itemize}
  \item interchanging the $i$-th and $j$-th columns;
  \item multiplying the $j$-th column on the right by  $a\in \bar{\mathcal{R}}$ and adding it to the $i$-th column;
  \item multiplying the $i$-th column on the right by a nonzero element $\alpha\in \mathcal{F}$;
\end{itemize}
and the same operations on rows with right replaced by left.

Every elementary operation can be represented by right or left multiplication by an elementary unimodular matrix. Furthermore, every unimodular matrix $\mathbf{U}$ or $\mathbf{V}$  can be obtained as a product of such elementary unimodular matrices. The algorithm transforming each matrix to the normal form resembles the Gauss elimination procedure; the details can be found in \cite{Cohn}.

In practical terms, one can use the  decomposition (\ref{JNF}) for constructing a basis in a finitely generated $\bar{\mathcal{R}}$-module.  Given a right $\bar{\mathcal{R}}$-module $\mathcal{M}\subset \bar{\mathcal{R}}^n$ generated by the vectors $M_1,\ldots , M_m$, define the matrix $\mathbf{M}=(M_1,\ldots, M_m)\in \bar{\mathcal{R}}^{n\times m}$. According to (\ref{JNF}) the module $\mathcal{M}$ is freely generated by the first $l-1$ columns of the matrix $\mathbf{U}^{-1}$ plus the $l$-th column of $\mathbf{U}^{-1}$ multiplied by $\Delta$ on the right.

\end{appendix}

\end{document}